\documentclass[11pt,letterpaper]{article}

\usepackage[margin=1in,letterpaper]{geometry}
\usepackage[colorlinks,urlcolor=blue,citecolor=blue,linkcolor=blue]{hyperref}
\usepackage{fullpage}
\usepackage{subcaption}
\usepackage{graphicx}
\usepackage{epstopdf}
\graphicspath{{./figs/}}

\usepackage{multirow}
\usepackage{comment}
\usepackage{amsfonts,amsmath,mathtools,amssymb,amsthm}
\usepackage{xcolor}
\usepackage[linesnumbered,ruled,vlined]{algorithm2e}

\usepackage{bbm,bm}
\usepackage{xspace,prettyref}
\usepackage{color}
\usepackage{dsfont}

\SetKwInput{KwInput}{Input}                
\SetKwInput{KwOutput}{Output}              

\title{Testing Triangle Freeness in the General Model in Graphs with Arboricity $O(\sqrt{n})$}

\author{
Reut Levi\thanks{ Email: {\tt reut.levi1@idc.ac.il}.}
}

\newcommand{\eps}{\epsilon}
\newcommand{\eqdef}{\stackrel{\rm def}{=}}

\newtheorem{theorem}{Theorem}
\newtheorem{lemma}{Lemma}
\newtheorem{claim}[lemma]{Claim}
\newtheorem{definition}[lemma]{Definition}

\newtheorem{remark}{Remark}

\newcommand{\poly}{{\rm poly}}
\newcommand{\E}{{\rm E}}

\renewcommand{\Pr}{\mathrm{Pr}}

\begin{document}

\maketitle

\begin{abstract}
We study the problem of testing triangle freeness in the general graph model. 
This problem was first studied in the general graph model by Alon et al. (SIAM J. Discret. Math. 2008) who provided both lower bounds and upper bounds that depend on the number of vertices and the average degree of the graph.
Their bounds are tight only when $d_{\rm max} = O(d)$ and $\bar{d} \leq \sqrt{n}$ or when $\bar{d} = \Theta(1)$, where $d_{\rm max}$ denotes the maximum degree and $\bar{d}$ denotes the average degree of the graph.
In this paper we provide bounds that depend on the arboricity of the graph and the average degree.
As in Alon et al., the parameters of our tester is the number of vertices, $n$, the number of edges, $m$, and the proximity parameter $\epsilon$ (the arboricity of the graph is not a parameter of the algorithm).
The query complexity of our tester is $\tilde{O}(\Gamma/\bar{d} + \Gamma)\cdot \poly(1/\eps)$ on expectation, where $\Gamma$ denotes the arboricity of the input graph (we use $\tilde{O}(\cdot)$ to suppress $O(\log \log n)$ factors). 
We show that for graphs with arboricity $O(\sqrt{n})$ this upper bound is tight in the following sense. 
For any $\Gamma \in [s]$ where $s= \Theta(\sqrt{n})$ there exists a family of graphs with arboricity $\Gamma$ and average degree $\bar{d}$ such that $\Omega(\Gamma/\bar{d} + \Gamma)$ queries are required for testing triangle freeness on this family of graphs.
Moreover, this lower bound holds for any such $\Gamma$ and for a large range of feasible average degrees~\footnote{For a graph, $G$, whose arboricity is $\Gamma$, the number of edges is at most $n \cdot \Gamma$ and at least $\Gamma^2$. Thus, the average degree of $G$ is at least $\Gamma^2/n$ and at most $\Gamma$.}.
\end{abstract}


\section{Introduction}
Testing triangle-freeness is one of the most basic decision problems on graphs. The existence of triangles in a graph is often a crucial property for various applications. 
In the realm of property testing, decision problems are relaxed so that a tester for a property $\mathcal{P}$ is only required to distinguish between graphs that have the property $\mathcal{P}$ from graphs which are ``far'' according to some predetermined distance measure, from having the property $\mathcal{P}$, which in our case are graphs which are far from being triangle free. 

Testing triangle freeness is known to be possible with query complexity which only depends on the proximity parameter, $\eps$, in graphs which are either dense or sparse. 
More specifically, Alon, Fischer, Krivelevich and Szegedy~\cite{AFKS00} showed that in the dense-graphs model~\cite{dense} it is possible to test triangle-freeness with query complexity which is independent of the size of the graph but has a tower-type dependence in $1/\eps$.
In the other extreme, Goldreich and Ron~\cite{GR02} observed that in the bounded-degree model~\cite{GR02} it is possible to test triangle-freeness with query complexity $O(1/\eps)$ given that the maximum degree of the input graph is constant.

Alon, Kaufman, Krivelevich, Ron~\cite{AKKR08} were the first to study this problem in the general-graphs model~\cite{ParnasR02, KKR04}.
This model is more stringent in the sense that we do not assume anything on the density of the graph and the distance is measured with respect to the actual number of edges in the graph (instead of the maximum possible number of edges). 
They provided several upper bounds which apply for almost the entire range of average degrees. They also provided lower bounds that show that their upper bounds are at most quadratic in the optimal bounds.  
Shortly after, Rast~\cite{Ras06} and Gugelmann~\cite{Gug06} improved their upper bounds and lower bounds, respectively, for some ranges of the parameters. 

Although there is a fairly significant gap between the known upper bounds and lower bounds for the vast range of parameters, there has been no progress on this question since then.
In this paper we provide an upper bound and several lower bounds which are tight for a large range of parameters.
Surprisingly, our bounds depend on the arboricity of the graph although it is not a parameter of our algorithm.

\subsection{Results}
We provide an upper bound whose running time complexity is $\tilde{O}(\Gamma/\bar{d} + \Gamma) \cdot \poly(1/\eps)$ on expectation.
Therefore, for $m \leq n$ our upper bound is $\tilde{O}(\Gamma/\bar{d})$ and when $m > n$ our upper bound is $\tilde{O}(\Gamma)$ (ignoring polynomial dependencies in $1/\eps$). 

We provide three lower bounds, each suitable for a different range of parameters. 
\begin{enumerate}
\item For any $\Gamma$ and any feasible $m \geq 1$, we provide a lower bound of $\Omega((\Gamma n)/m) = \Omega(\Gamma/\bar{d})$ queries.
Therefore our upper bound is tight when $m\leq n$ (up to polynomial dependencies in $1/\eps$ and $O(\log \log n)$ factors).

\item For any $\Gamma$ and any feasible $m \geq \Gamma^3$ we provide a lower bound of $\Omega(\Gamma)$ queries.
Therefore, our upper bound is also essentially tight as long as $m \geq \Gamma^3$ (notice that since $m \leq n\cdot \Gamma$, it is implied that this lower bound applies only for graphs in which $\Gamma = O(n^{1/3})$).
Since we may assume that $m\geq n$ (otherwise we already have essentially tight lower bound), one implication of this lower bound is that our upper bound is tight in the strong sense for graphs with arboricity $O(n^{1/3})$ (namely it is tight for any feasible $m$) as it is always the case that $m \geq \Gamma^3$ for $\Gamma = O(n^{1/3})$. 

\item For any $\Gamma \leq (n/2)^{1/2}$ and any feasible $n \leq m \leq \Gamma^3$ we provide a lower bound of $\Omega(m^{1/3})$ queries. 
Since it is always the case that $m \geq \Gamma^2$, a lower bound of $\Omega(\Gamma^{2/3})$ queries is also implied. 
\end{enumerate}

To summarize, for graphs of arboricity $\Gamma = O(n^{1/2})$ we obtain that our upper bound is tight for a large range of average degrees. 
Additionally, for the range of average degrees in which we do not provide tight bounds, our upper bound is essentially $O(\Gamma)$ while our lower bound is $\Omega(\Gamma^{2/3})$ in the worst case.  

\subsection{Related Work}
\subsubsection{Property testing of triangle freeness}
\sloppy
As mentioned above, testing triangle freeness, in the context of property testing, was first studied by Alon et al.~\cite{AFKS00} in the dense graphs model.
They showed that triangle freeness can be tested in time which is independent of the size of the graph. 
However, their upper bound has tower-type dependence in $1/\eps$.
Alon~\cite{Alo02} showed that the query complexity of this problem in the dense-graphs model is indeed super-polynomial in $1/\eps$. 

In the bounded degree model Goldreich and Ron~\cite{GR02} observed that it is possible to test triangle freeness with query complexity $O(1/\eps)$ in graphs of maximum degree bounded by some constant. 

The problem of  testing triangle freeness in the general graph model was first studied by Alon, Kaufman, Krivelevich, Ron~\cite{AKKR08}.
The query complexity of their algorithms dependent on $n$ and $\bar{d}$, the number of vertices in the graph and the average degree, respectively.  
They provided sublinear upper bounds for almost the entire range of parameters. Moreover, their upper bounds are at most quadratic in their lower bounds.
Specifically, their upper bound, which is combined from several upper bounds is $\tilde{O}(\min \{ (n\bar{d})^{1/2}/\eps^{3/2}, (n^{4/3}/\bar{d}^{2/3})/\eps^2\})$.
Their lower bound, which is also combined from several lower bounds, is $\Omega(\max\{(n/\bar{d})^{1/2}, \min\{\bar{d}, n/\bar{d}\}, \min\{\bar{d}^{1/2}, n^{2/3}/\bar{d}^{1/3}\}\cdot n^{-o(1)}\})$. 

\sloppy 
Rast~\cite{Ras06} improved the upper bound of~\cite{AKKR08} for graphs with average degree in the range $[c_1 n^{1/5}, c_2 n^{1/2}]$ where $c_1$ and $c_2$ are some constants.  
The upper bound in~\cite{Ras06} is $O(\max\{(n\bar{d})^{4/9}, n^{2/3}/\bar{d}^{1/3}\})$.

Gugelmann~\cite{Gug06} provided a lower bound which improves the lower bound in~\cite{AKKR08} for graphs with average degree in the range $[c_1 n^{2/5}, c_2 n^{4/5}]$ where $c_1$ and $c_2$ are some constants. 
The lower bound in~\cite{Gug06} is $\Omega(\min\{(n\bar{d})^{1/3}, n/\bar{d}\})$.

\subsubsection{Sublinear algorithms that receive the arboricity of the graph as a parameter}
Eden, Ron and Rosenbaum~\cite{ERR19} designed an algorithm that given $n$, the number of edges of the input graph and an upper bound on the arboricity of the input graph, $\Gamma$, 
the algorithm makes $O(\Gamma/\bar{d} + \log^3 n/\eps)$ queries on expectation and samples an edge of the graph almost uniformly. 
More specifically, each edge in the graph is sampled with probability in the range $[(1-\eps)m, (1+\eps)m]$. 
 
Eden, Ron and Seshadhri~\cite{ERS19} estimate the degree distribution moments of an undirected graph.
In particular, for estimating the average degree of a graph, their algorithm has query complexity of $\tilde{O}(\Gamma/\bar{d})$. 
As they show in their paper, if $\Gamma$ is not given as an input to the algorithm then estimating the average degree is not possible in general with this query complexity. 

In another paper, Eden, Ron and Seshadhri~\cite{ERS20} give a $(1\pm \eps)$-approximation for the number of $k$-cliques in a graph given a bound on the arboricity of the graph $\Gamma$. 
In particular for triangles they provide an upper bound with expected running time, in terms of $n$, $\Gamma$ and the number of triangles in the graph, $n_3$, of $\min\{n\Gamma^2/n_3, n/n_3^{1/3} + (m\Gamma)/n_3\}\cdot \poly(\log n, 1/\eps)$.

\subsubsection{Testing graphs for bounded arboricity}
Eden, Levi and Ron~\cite{ELR20} provided an algorithm for testing whether a graph has bounded arboricity. Specifically, they provide a tolerant tester that distinguished graphs that are $\eps$-close to having arboricity $\Gamma$ from which are $c\cdot \eps$-far from having arboricity $3 \Gamma$, where $c$ is an absolute constant.  
The query complexity and the running time of their algorithm is in terms of $n$, $m$ and $\Gamma$ is $\tilde{O}(n/\sqrt{m} + n\Gamma/m)$ and is quasi-polynomial in $1/\eps$.

\subsection{Comparison between our upper bound and upper bounds in previous work}
As mentioned before, Alon et al.~\cite{AKKR08} provide tight bounds only in two cases. 
The first case is when $d_{\rm max} = O(\bar{d})$ and $\bar{d} \leq \sqrt{n}$, where $d_{\rm max}$ denotes the maximum degree and $\bar{d}$ denotes the average degree of the graph.
In this case, it follows that $\Gamma = \Theta(d_{\rm max}) = \Theta(\bar{d})$ an so our upper bounds essentially match.
Additionally we note that a bound on the arboricty of the graph does not imply a bound on the maximum degree of the graph. 
In fact, the maximum degree could be $\Theta(n)$ while the arboricity is $\Theta(1)$ (as it is the case in the star graph).
Consequently, the tightness of our upper bound is not restricted for graphs which have bounded maximum degree. 

The second case is when $\bar{d} = \Theta(1)$, for these graphs the running time complexity of their algorithm is $\tilde{\Theta}(n^{1/2})$.
For this case, the running time complexity of our upper bound is $\tilde{O}(\Gamma)$. We note that in graphs in which $\bar{d} = \Theta(1)$, $\Gamma$ could range between $\Theta(1)$ and $\Theta(n^{1/2})$. Therefore when $\bar{d} = \Theta(1)$ the complexity of our upper bound is not worse than the complexity of the upper bound in~\cite{AKKR08} but could be much better, depending on $\Gamma$.

For average degree in the range between $\Omega(1)$ and $O(n^{2/5})$ and in the range between $\Omega(n^{1/2})$ and $O(n^{2/3})$ the upper bound of $O(m^{1/2})$ queries of Alon et al. achieves the best running time, in terms of $n$ and $m$. 
For these ranges, the running time of our algorithm is $\tilde{O}(\Gamma)$. Since $m^{1/2} \geq \Gamma$, we obtain that for this ranges as well the performances of our algorithm are at least as good (up to $O(\log \log n)$ and $\poly(1/\eps)$ factors) but could be significantly better.
 
For average degree in the range between $\Omega(n^{2/5})$ and $O(n^{1/2})$ the upper bound of $O(\max\{(n\bar{d})^{4/9}, n^{2/3}/\bar{d}^{1/3}\})$ queries of Rast~\cite{Ras06} achieves the best running time. In this range our upper bound is always better than the upper bound of~\cite{Ras06} for graphs of arboricity $O(n^{12/21})$.

\subsection{High-level of Our Algorithm}
It is well known that a graph which is $\eps$-far from being triangle free has $\Omega(\eps m)$ edge-disjoint triangles (see Claim~\ref{clm:trg}).
Therefore if we were able to sample edges uniformly from the graph then after sampling $O(1/\eps)$ edges we would sample an edge $\{u, v\}$ which belongs to a triangle.
Thus, if we revealed the entire neighborhood of $u$ and the entire neighborhood of $v$ then we would find a triangle in the graph. 
Our algorithm is based on this simple approach. 
There are only two problems that need to be addressed.
The first problem is that sampling edges uniformly in a graph in which the degrees have high variability is too costly. 
The second problem, which also stems from the variability of the degrees in the graph, is that revealing the entire neighborhood of a vertex can be too costly, depending on its degree. 

This is where the arboricity of the graph comes into play.
For a graph of arboricity $\Gamma$, as was shown in~\cite{ELR20}, the fraction of edges in the subgraph induced on {\em heavy} vertices, that is, vertices with degree greater than $c\Gamma/\eps$ where $c$ is some absolute constant, is at most $\eps/2$.
Therefore, if $\Gamma$ was given to us as a parameter then we could, in some sense, ignore the subgraph induced on vertices of degree greater than $\Theta(\Gamma/\eps)$ 
since a graph which is $\eps$-far from being triangle free still have $\Omega(\eps m)$ violating edges (and $\Omega(\eps m)$ edge-disjoint triangles) even after we remove this subgraph entirely.
Ignoring this subgraph allows us on one hand to sample edges almost uniformly from the resulting graph while making only $\Omega(\Gamma/\bar{d})$ queries, and also guarantees that there are $\Omega(\eps m)$ violating edges for which both endpoints are not heavy. This solves the two problems we had with taking the simple approach.
 
However a bound on the arboricty of the graph is not given to the algorithm as a parameter.
Since approximating the arboricity of a graph up to a constant factor is not possible in sublinear time (to see this consider a graph with a hidden clique), we estimate a different parameter which we informally refer to as the {\em effective arboricity} of the graph.
We show that this parameter suffices for our needs. In fact, this parameter could be much smaller than $\Gamma$, in which case the complexity of our algorithm is better than $O(\Gamma/\bar{d} + \Gamma)$.          
We reduce the problem of approximating the effective arboricity of the graph to the problem of estimating the number of edges in the graph in which we remove the subgraph induced on heavy vertices, where heavy vertices are defined with respect to increasing thresholds. We stop increasing our threshold once the estimation of the number of edges is sufficiently large. 
As we prove, with high constant probability, our approximation to the effective arboricity is bounded by $O(\Gamma)$ which leads to a tester with running time $O(\Gamma/\bar{d} + \Gamma)$, as claimed.    


\subsection{Lower Bounds}
Our first lower bound of $\Omega(\Gamma/\bar{d})$ for graphs in which $\bar{d} \leq 1$ is based on a simple hitting argument. 
Specifically, construct a graph which is $1/3$-far from being triangle free in which $\Omega(\Gamma/\bar{d} + \Gamma) =  \Omega(\Gamma/\bar{d})$ queries are required in order to sample a vertex which is not isolated with probability that is at least $1/3$.

Our other two lower bounds are simple adaptations of the lower bound of $\Omega(\min\{\bar{d}, n/\bar{d}\})$ queries presented in Alon et al.~\cite{AKKR08}.

\section{Preliminaries}
Let $G = (V, E)$ be an undirected graph and let $\bar{d} = 2m/n$ denote the average degree of $G$ where $n = |V|$ and $m = |E|$.  
For each vertex $v\in V$, let ${\rm deg}(v)$ denote the number of neighbors of $v$. 
For a subset of vertices $S \subseteq V$ we denote by $G([S])$ the subgraph induced on $S$.
For a directed graph $D$ we denote by $\bar{d}_{out}(D)$ the average out-degree of $D$.

A graph $G$ is triangle free if for every three vertices, $u, v, w$ in $G$ at least one pair in $\{ \{u, v\}, \{v, w\}, \{w, u\}\}$ is not an edge of $G$.
A graph $G$ is $\eps$-far from being triangle free if more than $\eps m$ edges need to be removed in order to make $G$ triangle free.

In the general graph model the tester accesses the graph via the following oracle queries.
\begin{enumerate}
\item Degree queries: on query $v$ the oracle returns ${\rm deg}(v)$.
\item Neighbor queries: on query $(v, i)$ where $i \in [{\rm deg}(v)]$, the oracle returns the $i$-th neighbor of $v$.
\item Vertex-pair queries: on query $\{v, u\}$ the oracle returns whether there is an edge between $u$ and $v$.
\end{enumerate}

An algorithm is a tester for the property of triangle freeness if given a proximity parameter $\eps$ and access to an input graph $G$, it accepts $G$ with probability at least $2/3$ if $G$ is triangle free and 
rejects $G$ with probability at least $2/3$ if $G$ is $\eps$-far from being triangle free. If the tester always accepts graphs which are triangle free we say it has one-side error. Otherwise we say it has two-sided error.   

\begin{claim}\label{clm:trg}
A graph $G = (V, E)$ which is $\eps$-far from being triangle free has at least $\eps m / 3$ edge-disjoint triangles.
\end{claim}
\begin{proof}
Consider a procedure that given a graph $G$, as long as there is triangle, $t$ in $G$ it deletes all the edges of $t$ and proceeds in this manner until there are no triangles in the graph.
The number of edges which are deleted by this process is at least $\eps m$. Therefore the number of edge disjoint triangles that are deleted is at least $\eps m /3$. The claim follows.  
\end{proof}

The {\em arboricity} of an undirected graph $G$, denoted by $\Gamma(G)$, is the minimum number of forests into which its edges can be partitioned. Equivalently it is the minimum number of spanning forests needed to cover all the edges of the graph.

\section{The Algorithm}

\subsection{First Step: computing the threshold for defining heavy vertices}

As described above, for an input graph $G$, the number of edges in the subgraph induced on the heavy vertices w.r.t. the threshold $4\Gamma(G)/\eps$ is at most $(\eps/2) |E(G)|$ (see Claim~\ref{clm:arb}).
Therefore, when testing triangle freeness, we may, roughly speaking, ignore this subgraph with the hope of obtaining better complexity. 
Since $\Gamma(G)$ is not given to the algorithm as a parameter, we compute, in Algorithm~\ref{pro:alpha}, a different parameter of the graph, denoted by $\Gamma^*$, which is, roughly speaking, an approximation of the {\em effective} arboricity of the graph. 
In order to specify the guarantees on $\Gamma^*$ we shall need a couple of definitions.

\begin{definition}
For a graph $G = (V, E)$ and a threshold $t$ we define the set of {\em heavy vertices with respect to $t$} as $H_t(G) = \{v\in V: d(v) > t\}$ 
and the set of {\em light vertices with respect to $t$} as $L_t(G) = V \setminus H_t(G)$.
\end{definition}
When $G$ is clear from the context we may simply use $H_t$ and $L_t$. 
Using the definition of $H_t(G)$, we next define the graph $H(G, t)$ which is defined w.r.t. $G$ and a threshold $t$.

\begin{definition}[\bf The undirected graph $H(G, t)$]
For a graph $G = (V, E)$ and a threshold $t$, the graph $H(G, t)$ is an undirected graph defined as follows.
The set of vertices of $H(G, t)$ is $V$ and the set of edges of $H(G, t)$ is $E(G)\setminus \{\{u, v\} : u\in H_t(G) \text{ and } v \in H_t(G)\}$.
Namely, $H(G, t)$ is the graph $G$ after removing the edges for which both endpoints are heavy with respect to $t$. 
\end{definition}

\begin{claim}\label{clm:arb}
For a graph $G$, $\Gamma' \geq \Gamma(G)$ and $\eta\in (0,1]$ it holds that $|E(H(G, \Gamma'/\eta))| \geq (1-2\eta)|E(G)|$.
\end{claim}
\begin{proof}
We shall prove the claim about $\Gamma' = \Gamma$. 
The general claim will follow from the fact that $|E(H(G, x)|$ is monotonically non-decreasing in $x$. 
Let $G([H_t])$ denote the sub-graph induced on $H_t(G)$ where $t = \Gamma/\eta$ and let $k$ denote the number of edges of this graph.
Our goal is to show that $k < 2\eta m $ where $m = |E(G)|$.
The sum of the degrees of vertices in $H_t(G)$ is greater than $t \cdot |H_t(G)|$, therefore $m > t\cdot |H_t(G)|/2$.
On the other hand, since the arboricity of $G([H_t])$ is also bounded by $\Gamma$ it follows that $k \leq |H_t(G)| \cdot \Gamma$.
Therefore $k < 2m/t \cdot \Gamma = 2\eta m$, as desired.
\end{proof}

\medskip
The guarantees on $\Gamma^*$, which is the return value of Algorithm~\ref{pro:alpha} (that will be described next), are as described in the following claim.

\begin{claim}\label{lem:alpha}
With probability at least $5/6$, $\Gamma^*$ returned by Algorithm~\ref{pro:alpha} is such that:
\begin{enumerate}
\item $|E(H(G, t))| \geq (1-(\eps/6))m$ where $t = \Gamma^*/(48\epsilon)$,
\item $\Gamma^* \leq 2 \Gamma(G)$.
\end{enumerate}
\end{claim} 

\medskip
Algorithm~\ref{pro:alpha} proceeds in iterations where in each iteration it multiplies $\Gamma^*$ by a factor of $2$, where initially $\Gamma^*$ is set to $1$. 
It stops when the estimated number of edges in $E(H(G, t))$, for $t$ which is $\Theta(\Gamma^{*}/\eps)$, is at least $|E(G)| (1-\Theta(\eps))$.
In order to estimate the number of edges in $E(H(G, t))$, Algorithm~\ref{pro:alpha} calls Algorithm~\ref{alg:est}.

In turn, Algorithm~\ref{alg:est} uses the directed graph $D(G, t)$, which we defined momentarily, that is constructed from $G$ and can be accessed by making a constant number of queries to $G$. 

\begin{definition}[\bf The directed graph $D(G, t)$]
The graph $D(G, t)$ is a directed version of the graph $H(G, t)$ in which we orient the edges as follows.
For every edge $\{u , v\}$ of $H(G, t)$, we orient the edge from $u$ to $v$ if:
(a) $u \in L_t$ and $v \in H_t$ or (b) both $u \in L_t$ and $v \in L_t$ and $id(u) < id(v)$.
Otherwise, we orient it from $v$ to $u$. 
\end{definition}

\begin{algorithm}[ht]
\DontPrintSemicolon
\KwInput{Access to a graph $G$ and parameters $n$, $m$ and $\eps \in (0, 1]$}
\KwOutput{$\Gamma^*$ as described in Lemma~\ref{lem:alpha}}
Set $\Gamma_1 = 1$\; 
\For{$i = 1$ \textbf{to} $\log n$} {
Run Algorithm~\ref{pro:est} on $G$ with parameters $n$, $m$, $\eps/24$, $t_i$ and $\delta = \Theta(\log \log n)$ where $t_i \eqdef  \Gamma_i/(24\eps)$ \label{step1e}. Let $Z_i$ denote the returned value.\;
If $Z_i \leq (1 - \eps/12) m$ then set $\Gamma_{i+1} = 2 \Gamma_{i}$, otherwise, return $\Gamma_i$\label{step:alpha4}\;
}
\caption{Compute $\Gamma^*$}\label{pro:alpha}
\end{algorithm}

\begin{algorithm}[ht]
\DontPrintSemicolon
\KwInput{Access to an undirected graph $G$ and parameters parameters $n$, $m$, $\eps$, $t$ and $\delta$, where $n = |V(G)|$ and $m = |E(G)|$}
\KwOutput{Estimation to the number of edges of $H(G, t)$}
Sample $r = \Theta(\delta^{-1}\eps^{-2} t/\bar{d})$, where $\bar{d} = m/n$, vertices $v_1, \ldots, v_r$, uniformly at random from $V(G)$.\;
For each $i\in [r]$:\;
\begin{enumerate}
\item Sample a random neighbor of $v_i$, $u$. If the edge between $u$ and $v_i$ is oriented from $v_i$ to $u$ in $D$ then set $Y_i = 1$, otherwise set $Y_i = 0$\label{alg1:step1}
\item Set $X_i = \frac{(d_{\rm out}(v) + d_{\rm in}(v)) \cdot Y_i}{t}$\label{alg1:step2}
\end{enumerate}
Return $X = \frac{t \cdot |V(G)|}{r} \cdot \sum_{i \in [r]} X_i$\;
\caption{Estimate the number of edges of $H(G, t)$}\label{pro:est}
\end{algorithm}\label{alg:est}

The following claim specifies the guarantees of Algorithm~\ref{alg:est}.

\begin{claim}\label{clm:est}
Given a query access to a graph $G$ and parameters $n$, $m$, $\eps$, $t$ and $\delta$ where $n = |V(G)|$ and $m = |E(G)|$, Algorithm~\ref{pro:est} outputs $X$ such that w.p. at least $1- 2^{1/\delta}$,
$(1-\eps)m' \leq X \leq (1+\eps)m'$ if $m' \geq (1-2\eps)m$, and $X < (1-\eps)m$, otherwise, where $m' = |E(H(G, t))|$.
\end{claim}
\begin{proof}
First observe that $Y_i$ is an indicator variable to the event that the edge selected in the $i$-th iteration of Algorithm~\ref{pro:est}, $\{v_i, u\}$ is an out-edge of $v_i$ in $D(G, t)$.
Since $V(G) = V(D(G,t))$ we obtain the following. 
\begin{equation}
\E(Y_i) = \frac{1}{|V(D(G, t))|} \cdot \Sigma_{v\in V(D(G, t))} \frac{d_{\rm out}(v)}{d_{\rm out}(v) + d_{\rm in}(v)}\;.
\end{equation} 
Similarly,
\begin{equation}
\begin{split}
\E(X_i) = &  ~\frac{1}{|V(D(G, t))|} \cdot \Sigma_{v\in V(D(G, t))} \frac{d_{\rm out}(v) + d_{\rm in}(v)}{t} \cdot \frac{d_{\rm out}(v)}{d_{\rm out}(v) + d_{\rm in}(v)} \\
= &  ~\frac{1}{|V(D(G, t))|} \cdot \Sigma_{v\in V(D(G, t))}  \frac{d_{\rm out}(v)}{t} = \frac{\bar{d}_{\rm out}(D(G, t))}{t}\;.
\end{split}
\end{equation}
Observe that if $v_i \in H_t(G)$ then $d_{\rm out}(v) = 0$ and so $Y_i = X_i= 0$.
On the other hand, if $v_i \in L_t(G)$ then $d_{\rm out}(v) + d_{\rm in}(v) \leq t$. 
Therefore, in both cases $X_i \in [0, 1]$.
Thus, for $r$ which is $\Theta(1/(\delta\eps^{2}E(X_1)))$ it follows by Multiplicative Chernoff's bound (see Theorem~\ref{thm:chermul} in the appendix) 
 that with probability at least $1-2^{1/\delta}$, 
$$(1-\epsilon) E(X_1)  \leq \sum_{i\in [r]} {X_i}/r \leq (1+\epsilon) E(X_1)\;.$$
And so 
$$t \cdot |V(G)| \cdot (1-\epsilon)  E(X_1) \leq X \leq  t \cdot |V(G)| \cdot (1+\epsilon) E(X_1)\;.$$
Since $t \cdot |V(G)| \cdot E(X_1) = |E(H(G, t))|$ we obtain that 
$$(1-\eps) \cdot |E(H(G, t))| \leq X \leq (1+\eps) \cdot |E(H(G, t))|\;.$$

Hence, if $|E(H(G, t))| \geq (1-2\eps)m$ then $\bar{d}_{\rm out}(D(G, t)) = \Theta(\bar{d}(G))$ and so $\E(X_1) = \Theta(\bar{d}(G)/t)$, implying that $r = \Theta(1/(\delta\eps^{2}E(X_1)))$, as desired.
On the other hand, if $|E(H(G, t))| < (1-2\eps)m$, then it is not hard to see that the claim follows by a straightforward coupling argument. 
More specifically, first assume that $|E(H(G, t))| = (1-2\eps)m$ and so by the above, with probability at least $1-2^{1/\delta}$,
$$X \leq (1+\eps)|E(H(G, t))| = (1+\eps)(1-2\eps)m < (1-\eps)m$$.
Therefore, it follows by a coupling argument that with probability at least $1-2^{1/\delta}$, $X < (1-\eps)m$ also in the case that $|E(H(G, t))| < (1-2\eps)m$.
\end{proof}

\begin{claim}
For an input graph $G$ and parameters $n$, $m$, $\eps$, $t$ and $\delta$, where $n = |V(G)|$ and $m = |E(G)|$, the time complexity and query complexity of Algorithm~\ref{pro:est} is $O(\delta^{-1}\eps^{-2} t/\bar{d}(G))$.
\end{claim}
\begin{proof}
The claim follows from the fact that in order to implement Steps~2.\ref{alg1:step1} and~2.\ref{alg1:step2} of Algorithm~\ref{pro:est} the algorithm makes a constant number of queries to $G$. Specifically,  for each $v_i$ the algorithm either performs a single degree query (in case $v_i \in H_t(G)$ then $Y_i = X_i = 0)$ 
or a single adjacency-list query and $2$ degree queries in case $v_i  \in L_t(G)$ (the orientation of the edge $\{v_i, u\}$ can be determined by the degrees and ids of $v_i$ and $u$). 
The implementation of Step~2.\ref{alg1:step2} does not require additional queries as $d_{\rm out}(v_i) + d_{\rm in}(v_i) = d_G(v_i)$.
\end{proof}

We are now ready to prove Claim~\ref{lem:alpha}.

\begin{proof}[Proof of Claim~\ref{lem:alpha}]
For the sake of analysis assume that Algorithm~\ref{pro:alpha} performs all $\log n$ iterations of the for-loop.
Let $E_i$ denote the event that $Z_i$ is as claimed in Claim~\ref{clm:est}.
By Claim~\ref{clm:est}, for a fixed $i$, the probability that $E_i$ occurs is at least $1/(6\log n)$ for an appropriate setting of $\delta$. 
Therefore, by the union bound the probability that $E_i$ occurs for all $i\in [\log n]$ is at least $5/6$. 
From this point on we condition on the event that indeed $E_i$ occurs for all $i\in [\log n]$.

Let $m_i = |E(H(G, t_i))|$ for every $i\in [\log n]$.
Let $j$ denote the iteration in which Algorithm~\ref{pro:alpha} returns a value. 
By Step~\ref{step:alpha4} of Algorithm~\ref{pro:alpha}, $Z_j > (1 - \eps/12) m$.  
By Claim~\ref{clm:est}, it follows that $m_j \geq (1-\epsilon/6)m$, as desired (to see this note that if $m_j < (1-\epsilon/6)m$ then By Claim~\ref{clm:est}, $Z_j < (1-\eps/12)m$).

To prove the claim about $\Gamma^*$ we consider the minimum $j'\geq 1$, for which $2^{j'-1} \geq \Gamma$.
If $j < j'$ then clearly $\Gamma^* < \Gamma$, as desired. 
Otherwise, we claim that $j = j'$ (namely, that $\Gamma^* = 2^{j'-1}$) which implies that $\Gamma^* \leq 2\Gamma$, as desired.
To see this, first note that by Claim~\ref{clm:arb}, for any $\Gamma' \geq \Gamma$ and $\eta\in (0,1]$, $|E(H(G, \Gamma'/\eta))| \geq (1-2\eta)m$.
Therefore $m_j' \geq (1-\eps/24)m$.
Therefore, by Claim~\ref{clm:est}, $Z_{j'} \geq (1-\eps/24)m_j' \geq (1-\eps/24)^2m > (1-\eps/12)m$, which implies that the algorithms stops at the $j'$-th iteration.
\end{proof}

\subsection{Second Step: sampling edges almost uniformly given a threshold for heavy vertices}

Given a threshold $t$, Algorithm~\ref{pro:sample} samples an edge from $H(G, t)$ almost uniformly as described in the next claim.
\begin{claim}\label{clm:sample}
Algorithm~\ref{pro:sample} samples an edge from $H(G, t)$ such that for each edge, $e$, of $H(G, t)$, the probability to sample $e$ is in $\left[\frac{c_1}{m'}, \frac{c_2}{m'} \right]$, where $c_1$ and $c_2$ are absolute constants and $m' \eqdef |E(H(G, t))|$.
If $t$ is such that $m' \geq |E(G)|/2$ then the expected running time of the algorithm is $O(t/\bar{d}(G))$. 
\end{claim}

\begin{proof}
Consider a single iteration of the while loop of Algorithm~\ref{pro:sample}.
For an edge $e$ in $H(G, t)$ let $p(e)$ denote the probability that $e$ is returned in this iteration of Algorithm~\ref{pro:sample}.
If $e$ is an edge such that both endpoints are in $L_t(G)$, then $p(e) = \frac{2}{n} \cdot \frac{1}{t}$\;.
If $e$ is an edge such that one endpoint in $L_t(G)$ and the other endpoint is in $H_t(G)$, then $p(e) = \frac{1}{n} \cdot \frac{1}{t}$\;.
Therefore for any two edges $e_1$ and $e_2$ in $H(G, t)$ the probability that $e_1$ is picked by the algorithm is at most twice the probability that $e_2$ is picked by the algorithm.
Since Algorithm~\ref{pro:sample} only returns edges in $H(G, t)$, the claim about the probability to sample an edge follows.

For a fixed iteration of the while loop, the probability that the algorithm returns one of the edges of $E(H(G, t))$ is at least $m' \cdot \frac{1}{n} \cdot \frac{1}{t}$ which is at least $\bar{d}(G)/(2t)$ in the case that $m' \geq |E(G)|/2$.
Therefore in this case the expected number of iterations of the while loop is at most $(2t/\bar{d}(G))$. Hence the claim about the expected running time follows.
\end{proof}

\begin{algorithm}[ht]
\DontPrintSemicolon
\KwInput{Access to a graph $G$ and a parameter $t$.}
\While{}{
Pick u.a.r. a vertex $v$ from $V(G)$\label{step1}\;
Pick u.a.r. $j \in [t]$\;
If $v \in L_t(G)$ and $v$ has a $j$-th neighbor, $u$, then return $\{v, u\}$.\;
}
\caption{Sample an edge from $H(G, t)$ almost uniformly}\label{pro:sample}\label{alg:sample}
\end{algorithm}

\begin{remark}
We remark that Algorithm~\ref{alg:sample} is stated as a Las Vegas algorithm. 
Moreover, if $m < |E(G)|/2$ then we can not obtain from Claim~\ref{clm:sample} any bound on the expected running time of the algorithm.
However, we note that since $t$ and $\bar{d}(G)$ are known then we can set a timeout for the algorithm (specifically $c t/\bar{d}(G)$ for some constant $c$) and incorporate the event that we were forced to stop the algorithm in the failure probability of the tester.
\end{remark}

\subsection{Putting things together - the algorithm for testing triangle freeness}

Using Algorithms~\ref{pro:est} and~\ref{pro:sample} we are now ready to describe our tester (Algorithm~\ref{alg:testing}). 
\begin{algorithm}[ht]
\DontPrintSemicolon
\KwInput{Access to a graph $G$ and parameters $n$, $m$, and $\eps$.}
Execute Algorithm~\ref{pro:est} with parameters $n$, $m$ and $\eps$ and let $\Gamma^*$ denote the returned value\;
Let $t = \Gamma^*/\eps$\;
\For{$i = 1$ \textbf{to} $s = \Theta(\eps^{-1})$} {
Execute Algorithm~\ref{pro:sample} on $G$ with parameter $t$ and let $\{u, v\}$ denote the edge returned by the algorithm.\label{step:edge}\;
If both $u \in L_t(G)$ and $v \in L_t(G)$ then return {\rm REJECT} if $N(u) \cap N(v) \neq \emptyset$\;
}
Return {\rm ACCEPT}\;
\caption{Testing Triangle-Freeness}\label{alg:testing}
\end{algorithm}

Since the tester has one-sided error (it rejects only if it finds a witness for violation, i.e., a triangle) its correctness follows from the following claim.  
\begin{claim}
If $G$ is $\eps$-far from being triangle free then {\bf Test Triangle Freeness} finds a triangle with probability at least $2/3$. 
The expected running time of the algorithm is $\tilde{O}\left(\Gamma/\bar{d}(G) + \Gamma \right)\cdot \poly(\eps^{-1})$.
\end{claim}
\begin{proof}
Let $G = (V, E)$ be an input graph which is $\eps$-far from being triangle free. 
There are at least $\eps m /3$ edge-disjoint triangles in $G$ (see Claim~\ref{clm:trg}).
Let $E_1$ denote the event that for $\Gamma^*$ that is return by Algorithm~\ref{pro:alpha} it holds that $|E(H(G, t))| \geq (1-(\eps/6))m$ where $t = \Gamma^*/\eps$.
By Claim~\ref{lem:alpha} $E_1$ occurs with probability at least $5/6$.
Given that $E_1$ occurred, it follows that there are at least $(\eps/3)m - (\eps/6)m = (\eps/6)m$ edge-disjoint triangles in $H(G, t)$.
Let $\{t_1, \ldots, t_k\}$ be an arbitrary subset of these edge-disjoint triangles where $k \eqdef (\eps/6)m$.
By the definition of $H(G,t)$ it holds that for every edge $\{u, v\} \in E(H(G,t))$ either $u\in L_t(G)$ or $v \in L_t(G)$.
Thus, for every $i\in [k]$ the triangle $t_i$ includes an edge $\{x_i, y_i\}$ such that both $x_i$ and $y_i$ are in $L_t(G)$.
Therefore, there are at least $k$ edges in $H(G, t)$ such that if Algorithm~\ref{pro:sample} returns one of these edges in Step~\ref{step:edge} of Algorithm~\ref{alg:testing} then Algorithm~\ref{alg:testing} rejects.
For every $i\in [s]$, let $E_{2,i}$ denote the event that Algorithm~\ref{pro:sample} returns one of these edges in the $i$-th iteration of Algorithm~\ref{alg:testing}.

By Claim~\ref{clm:sample}, given that $E_1$ occurred, the probability that $E_{2, i}$ occurs (given that the algorithm did not return {\rm REJECT} before the $i$-th iteration) is at least $c\eps$ for some absolute constant $c$.
Therefore the probability that both $E_1$ and $E_{2, i}$ occur for some $i\in [s]$ is at least $2/3$ for an appropriate setting of $s$.
Thus the algorithm rejects with probability at least $2/3$ as desired. 
\end{proof}

\section{Lower Bounds}
\subsection{Lower bound of $\Omega((\Gamma n)/m)$ for any $\Gamma$ and any feasible $m \geq 1$}
\begin{theorem}
For any $\Gamma \leq n-1$ and any $m\geq 1$ which is feasible w.r.t. $\Gamma$, any algorithm for testing triangle-freeness must perform $\Omega((\Gamma n)/m)$ queries where $n$, $m$ and $\Gamma$ denote the number of vertices, the number of edges and the arboricity of the input graph, receptively.
This lower bound holds even if the algorithm is allowed two-sided error.
\end{theorem}
\begin{proof}
We consider the following graph $G$ over $n$ vertices, $m$ edges and of arboricity $\Gamma$. $V_1$ and $V_2$ are subsets of $2m/(3\Gamma)$ vertices each and $V_3$ is a subset of $\Gamma/2$ vertices. $V_1$, $V_2$ and $V_3$ are pairwise disjoint.
The edges of the graph are as follows. The sub-graph induced on $V_1$ and $V_3$ is a complete bipartite graph with $V_1$ on one side and $V_3$ on the other side.
Similarly the subgraph induced on $V_2$ and $V_3$ is also a complete bipartite graph.
Between $V_1$ and $V_2$ we take $\Gamma/2$ edge disjoint prefect matchings. Consequently the degree of every node in $V_1$ and $V_2$, as well as the arboricity of the graph, is exactly $\Gamma$. 
The number of edge disjoint triangles in the graph is at least $m/3$. 
To see this consider the following correspondence between an edge $\{u, v\}$ such that $u\in V_1$ and $v\in V_2$ and a triangle in the graph.
Let $i \in [\Gamma/2]$ denote the matching for which $\{u, v\}$ belongs to, then triangle that corresponds to $\{u ,v\}$ is $\{u, v, w_i\}$. 

Therefore the graph is $(1/3)$-far from being triangle free (if $t_1, \ldots, t_{m/3}$ are edge disjoint triangles of $G$ then we need to delete at least one edge per triangle in order to make $G$ triangle free).
The number of queries we need to make to hit either $V_1$ or $V_2$ is $\Omega((\Gamma n)/m)$. The number of queries we need to make to hit $V_3$ is $\Omega(n/\Gamma)$. Since $m = \Omega(\Gamma^2)$, we obtain a lower bound of $\Omega((\Gamma n)/m)$ queries in order to hit a vertex from $V_1\cup V_2 \cup V_3$.
Therefore unless the tester makes $\Omega((\Gamma n)/m)$ queries, it can not distinguish between $G$ and the empty graph. The theorem follows. 
\end{proof}

\subsection{Lower bound of $\Omega(\Gamma)$ for any $\Gamma$ and any feasible $m \geq \Gamma^3$}
We adapt the following lower bound of Alon et al.~\cite{AKKR08}. 
\begin{theorem}\label{thm:alon}
Any algorithm for testing triangle-freeness must perform $\Omega(\min\{\bar{d}, n/\bar{d}\})$ queries.
This lower bound holds even if the algorithm is allowed two-sided error and even for $d_{\rm max} = O(\bar{d})$.
\end{theorem}

Using Theorem~\ref{thm:alon}, we shall prove the following claims.

\begin{claim}\label{claim:alpha}
For any $\Gamma$ and any feasible $m$ w.r.t. $\Gamma$ such that $m\geq \Gamma^3$, any algorithm for testing triangle-freeness must perform $\Omega(\Gamma)$ queries, where $m$ and $\Gamma$ denote the number of edges and the arboricity of the input graph, respectively.
This lower bound holds even if the algorithm is allowed two-sided error.
\end{claim}

\begin{proof}
Assume towards contradiction that there exists an algorithm $\mathcal{A}$ for testing triangle-freeness that is allowed two-sided error and performs $o(\Gamma)$ queries even for input graphs for which $m\geq \Gamma^3$, 
where $m$ and $\Gamma$ denote the number of edges and the arboricity of the input graph of $\mathcal{A}$, respectively.
We will show that there exists an algorithm $\mathcal{B}$ for testing triangle-freeness (with two-sided error) for graphs in which $M/N = \Gamma$, $N = m/\Gamma$ and the maximum degree is $\Gamma$, whose query complexity is $o(\Gamma)$, where $M$ and $N$ denote the number of edges and the number of vertices of the input graph of $\mathcal{B}$, respectively.
This will contradict the lower bound in Theorem~\ref{thm:alon} as $\min \{M/N, N^2/M\} = \min \{\Gamma, m/\Gamma^2\} = \Gamma$, where the last inequality follows from the fact that $m \geq \Gamma^3$.

Let $m$, $n$ and $\Gamma$ be such that $m$ is feasible w.r.t. $\Gamma$ and $m \geq \Gamma^3$.
Therefore $\Gamma^3 \leq m \leq \Gamma n$. 
Let $G$ be a graph over $N$ vertices and $M$ edges for which $M/N = \Gamma$, $N = m/\Gamma$ and the maximum degree is $\Gamma$. 
Given such an input graph $G$, the algorithm $\mathcal{B}$ simulates $\mathcal{A}$ on another graph, $G'$, that will be described momentarily, and returns the output of $\mathcal{A}$ on $G'$.
The graph $G'$ is constructed from $G$ and has the following properties:
\begin{enumerate}
\item $G'$ has arboricity $\Gamma$. \label{item:1}
\item Any query on $G'$ can be answered by performing at most a single query to $G$\label{item:2}
\item If $G$ is triangle free then $G'$ is triangle free as well. \label{item:3}
\item If $G$ is $\eps$-far from being triangle free then $G'$ is $\eps$-far from being triangle free as well. \label{item:4}
\end{enumerate} 
$G'$ is simply the graph $G$ with $n-N$ isolated vertices (observe that $n \geq N$ since $m \leq \Gamma n$).
Since $M = m$, Item~\ref{item:4} follows. Items~\ref{item:1} and~\ref{item:3} follow from construction and the bound on the maximum degree of $G$.
Item~\ref{item:2} follows from the fact that any query to the graph $G'$ is can be answered  either by performing a single query the graph $G'$ or to without performing any query to $G$ (in case the algorithm queries the subgraph induced on the additional $n- N$ isolated vertices of $G'$).   

The completeness and soundness of algorithm $\mathcal{B}$ follows from the correctness of algorithm $\mathcal{A}$ and Items~\ref{item:3} and~\ref{item:4}, respectively.
The claim about the query complexity of algorithm $\mathcal{B}$ follows from Item~\ref{item:2} and the assumption on the query complexity of algorithm $\mathcal{A}$.
This completes the proof of the claim. 
\end{proof}

\subsection{Lower bound of $\Omega(m^{1/3})$ for any $\Gamma \leq (n/2)^{1/2}$ and any feasible $n \leq m \leq \Gamma^3$}

\begin{claim}
For any $\Gamma \leq (n/2)^{1/2}$ and any feasible $m$ w.r.t. $\Gamma$ such that $n \leq m \leq \Gamma^3$, any algorithm for testing triangle-freeness must perform $\Omega(m^{1/3})$ queries, where $m$, $n$ and $\Gamma$ denote the number of edges, number of vertices and the arboricity of the input graph, respectively.
This lower bound holds even if the algorithm is allowed two-sided error.
\end{claim}

\begin{proof}
The proof of this claim follows the same lines as the proof of Claim~\ref{claim:alpha}.
Assume towards contradiction that there exists an algorithm $\mathcal{A}$ for testing triangle-freeness that is allowed two-sided error and performs $o(m^{1/3})$ queries even for input graphs for which $n \leq m \leq \Gamma^3$, 
where $m$, $n$ and $\Gamma$ denote the number of edges, number of vertices and the arboricity of the input graph, respectively.
We will show that there exists an algorithm $\mathcal{B}$ for testing triangle-freeness (with two-sided error) for graphs in which $M = N^{3/2} = \Theta(m)$ and the maximum degree is $M/N$, whose query complexity is $o(M/N)$, where $M$ and $N$ denote the number of edges and the number of vertices of the input graph of $\mathcal{B}$, respectively.
This will contradict the lower bound in Theorem~\ref{thm:alon} as $\min \{M/N, N^2/M\} = M/N$, where the last inequality follows from the fact that $M = N^{3/2}$.

Let $m$, $n$ and $\Gamma$ be such that $m$ is feasible w.r.t. $\Gamma$ and $n \leq m \leq \Gamma^3$.
Let $G$ be a graph over $N$ vertices and $M$ edges for which $M = N^{3/2} = m/2$ and for which the maximum degree is $M/N$.
Given such an input graph $G$, the algorithm $\mathcal{B}$ simulates $\mathcal{A}$ on another graph, $G'$, that will be described momentarily, and returns the output of $\mathcal{A}$ on $G'$.
The graph $G'$ is constructed from $G$ and has the following properties:
\begin{enumerate}
\item $G'$ has arboricity $\Gamma$. \label{item:2.1}
\item Any query on $G'$ can be answered by performing at most a single query to $G$\label{item:2.2}
\item If $G$ is triangle free then $G'$ is triangle free as well. \label{item:2.3}
\item If $G$ is $\eps$-far from being triangle free then $G'$ is at least $\eps/2$-far from being triangle free as well. \label{item:2.4}
\end{enumerate} 
$G'$ is composed of the graph $G$, a complete bipartite graph $A$ over $2\Gamma$ vertices and a graph $I$ of $n - N - 2\Gamma$ isolated vertices.
Observe that $n - N \geq 2\Gamma$ since $\Gamma^2 \leq n/2$ and $N = M^{2/3} \leq \Gamma^2/2^{2/3} \leq n/2$.  
The graph $A$ has $\Gamma$ vertices on each side, $A_1$ and $A_2$ and $\Gamma^2$ edges.  
Item~\ref{item:4} follows from the fact that $M = m - \Gamma^2 \geq m - n/2 \geq m/2$. 
Observe that maximum degree of $G$ is at most $\Gamma$ since $M/N \leq m^{1/3} \leq \Gamma$.
Therefore, Items~\ref{item:1} and~\ref{item:3} follow from the fact that the arboricity of $A$ is $\Gamma$ and the bound on the maximum degree of $G$.
Item~\ref{item:2} follows from the fact that any query to the graph $G'$ is can be answered  either by performing a single query the graph $G'$ or to without performing any query to $G$ (in case the algorithm queries the subgraphs $A$ or $I$).   

The completeness and soundness of algorithm $\mathcal{B}$ follows from the correctness of algorithm $\mathcal{A}$ and Items~\ref{item:3} and~\ref{item:4}, respectively.
The claim about the query complexity of algorithm $\mathcal{B}$ follows from Item~\ref{item:2} and the assumption on the query complexity of algorithm $\mathcal{A}$.
This completes the proof of the claim. 
\end{proof}

\bibliography{refs}

\appendix
\section{Appendix}
\begin{theorem}[Multiplicative Chernoff's Bound]\label{thm:chermul}
Let $X_1, \ldots, X_n$ be identical independent random variables ranging in $[0, 1]$, and let $p = \E[X_1]$. Then, for every $\gamma \in (0, 2]$, it holds that
\begin{equation}
\Pr\left[ \left| \frac{1}{n} \cdot \sum_{i\in [n]} X_i - p \right| > \gamma \cdot p\right] < 2 \cdot e^{-\gamma^2 pn /4} \;.
\end{equation}  
\end{theorem}
\end{document}